\RequirePackage{fix-cm}

\documentclass[smallcondensed]{svjour3}
\smartqed  % flush right qed marks, e.g. at end of proof
\usepackage{amssymb,amsmath}
\usepackage{url}
\usepackage[hidelinks]{hyperref}
\usepackage{graphicx}

\makeatletter
\def\@seccntformat#1{\@ifundefined{#1@cntformat}%
   {\csname the#1\endcsname.\space}%    default
   {\csname #1@cntformat\endcsname}}%  enable individual control
\newcommand\section@cntformat{\Large\bf\thesection.\space}
\newcommand\subsection@cntformat{\large\bf\thesubsection.\space}
\newcommand\subsubsection@cntformat{\bf\thesubsubsection.\space}
\makeatother
\newcommand{\Z}{\mathbb{Z}}
\newcommand{\F}{\mathbb{F}}

\begin{document}

\title{Additive perfect codes in Doob graphs%
\thanks{This research is supported by National Natural Science Foundation of China (61672036),
Technology Foundation for Selected Overseas Chinese Scholar, Ministry of Personnel of China (05015133),
Excellent Youth Foundation of Natural Science Foundation of Anhui Province (No.1808085J20),
and the Program of fundamental scientific researches
of the Siberian Branch of the Russian Academy of Sciences No.I.1.1.
(No.0314-2016-0016).}
}


\author{Minjia~Shi
\and
Daitao~Huang
\and
\href{https://orcid.org/0000-0002-8516-755X}{Denis~S.~Krotov}
%\and
%\href{https://orcid.org/0000-0002-4078-8301}{Patrick~Sol\'e}
}

\authorrunning{M.~Shi, D.~Huang, D.~Krotov} % if too long for running head

\institute{M.~Shi \at
%Key Laboratory of Intelligent Computing Signal Processing,
%Ministry of Education, Anhui University,
%No.3 Feixi Road, Hefei, Anhui, 230039, China; \\
School of Mathematical Sciences, Anhui University, Hefei, Anhui, 230601,
China \\
              \email{smjwcl.good@163.com}           %  \\
           \and
           D.~Huang \at
             School of Mathematical Sciences,
             Anhui University, Hefei, Anhui, 230601, China\\
             \email{dthuang666@163.com}
             \and
            D.~S.~Krotov \at
            Sobolev Institute of Mathematics,
            pr. Akademika Koptyuga 4,
            Novosibirsk 630090, Russia \\
            \email{krotov@math.nsc.ru}
%                          \and
%             P.~Sol\'e \at
%             4CNRS/LAGA, University of Paris 8,
%             2 rue de la Libert\'e,
%             93 526 Saint-Denis, France \\
%             \email{sole@enst.fr}
}

\date{Received: 2018-06 / Accepted: 2018-11}

\maketitle

\begin{abstract}
The Doob graph $D(m,n)$ is the Cartesian product of $m>0$ copies of the Shrikhande graph and $n$ copies of the complete graph of order $4$.
Naturally, $D(m,n)$ can be represented as a Cayley graph on the additive group $(Z_4^2)^m \times (Z_2^2)^{n'} \times Z_4^{n''}$, where $n'+n''=n$.
A set of vertices of $D(m,n)$ is called an additive code if it forms a subgroup of this group.
We construct a $3$-parameter class of additive perfect codes in Doob graphs
and show that the known necessary conditions of the existence of additive $1$-perfect codes in $D(m,n'+n'')$ are sufficient.
Additionally, two quasi-cyclic additive $1$-perfect codes are constructed in $D(155,0+31)$ and $D(2667,0+127)$.
\end{abstract}

%=============================
%=============================
%=============================
\section{\Large Introduction}\label{s:intro}

Perfect codes are a fascinating structure in coding theory, which attracts attention again and again.
The existence of perfect codes has been studied for various metrics, in particular, for the Hamming metric \cite{Tiet:1973}, \cite{ZL:1973}. Generally we consider a distance-regular graph $G(V,E)$ due to the important role of perfect codes in distance-regular graphs. A $1$-perfect code in a graph $G(V,E)$ is a subset $C$ of $V$, which is an independent set such that every vertex in $V \backslash C$ is adjacent to exactly one vertex in $C$.

The Doob graph $D(m,n)$ is the Cartesian product of $m$ copies of the Shrikhande graph and $n$ copies of the complete graph of order $4$, where the Shrikhande graph is a strongly regular graph with $16$ vertices and $48$ edges with each vertex having degree $6$.
All $D(m,n)$ with the same value $2m+n$ have the same parameters as distance-regular graphs; the partial case $m=0$ corresponds to the $4$-ary Hamming graph.
In \cite{Kro:perfect-doob}, the author completely solved the problem of existence of linear $1$-perfect codes in Doob graphs (a linear code in Doob graph forms a module over the Galois ring $\mathrm{GR}(4^2)$)
and proposed an open problem about the additive $1$-perfect codes (an additive code forms a module over $\Z_4$).
In the current paper, we are aimed at showing that for arbitrary odd $\Delta \geq 3$, even $\Gamma$ and $n'' \in \{4,7,10,\ldots,2^{\Delta}-1\}$, there exists an additive $1$-perfect code in $D(m,n'+n'')$, where $m=\frac{2^{2\Delta+\Gamma}-2^{\Delta+\Gamma}-2n''}{6}$,  $n'=\frac{2^{\Gamma+\Delta}-1-n''}{3}$. In particular, we construct $2$ codes that are both $1$-perfect and quasi-cyclic in $D(m,0+n'')$. Together with the results in \cite{Kro:perfect-doob} for even $\Delta$, our construction solves the problem of existence of additive $1$-perfect codes in Doob graphs for all feasible parameters.

 The material is arranged as follows. The next section compiles the  background necessary to the forthcoming sections. Section 3 contains the main result of this paper. Three quasi-cyclic additive $1$-perfect codes are listed in Section $4$ (one of them was known before).  Section $5$ concludes the article, and points out some open problems.
%=============================
%=============================
%=============================
\section{\Large Preliminaries}\label{s:prel}
%=============================
\subsection{\large\bf Galois rings}\label{ss:galois}
Let $\mathbb{Z}$ denote the ring of integers, and let $\mathbb{Z}_p = \mathbb{Z} / p\mathbb{Z}$
denote the factor-ring of residue classes of $\mathbb{Z}$ modulo $p$.
If $\mathbb{M}$ is a ring or a module over a ring, then $\mathbb{M}^+$ denotes the additive group of $\mathbb{M}$.
If $h(x)$ is a basic irreducible polynomial of degree $m$ over $\Z_4$ and $\varsigma$ is a root $h(x)$, then any element in the residue class ring $\Z_4[x]/(h(x))$ can be written as $h_0+h_1\varsigma+\ldots+h_{m-1}\varsigma^{m-1}$, which could also be viewed as the vector $(h_0,h_1,\ldots,h_{m-1})$ over $\Z_4$, where $h_0,~h_1,~\ldots,~h_{m-1}$ run through $\Z_4$ independently. In fact, the map defined by
$$\phi: ~~~\Z_4[x]/(h(x)) \rightarrow \Z_4^m,$$
$$h_0+h_1\varsigma+\ldots+h_{m-1}\varsigma^{m-1}\mapsto  (h_0,h_1,\ldots,h_{m-1})$$
is a $\Z_4$-module isomorphism from $\Z_4[x]/(h(x))$ to $\Z_4^m$.
%Moreover, if $s$ is the minimum integer such that $s(h_0+h_1\varsigma+\cdots+h_{i-1}\varsigma^{i-1})=0$ if and only if $s$ is the minimum integer such that $s(h_0,h_1,\cdots,h_{i-1})=0$, which implies that the order of any element $h_0+h_1\varsigma+\cdots+h_{i-1}\varsigma^{i-1}$ in $\Z_4[x]/(h(x))$ is the same as the order of its corresponding element $(h_0,h_1,\cdots,h_{i-1})$ in $\Z_4^i$ \todo{I thing it is automatically implied by the isomorphism}.
 As usual, denote by Galois ring $\mathrm{GR}(4^m)$ the residue ring $\Z_4[x]/(h(x))$, and we denote by $\mathrm{GR}(4^m)^*$ the set of units of $\mathrm{GR}(4^m)$.
 The {\it Teichmuller} set $ \mathcal{T}=\{x \in \mathrm{GR}(4^m)|x^{2^{m}}=x\}$ is a set of representatives of the {\it residue field $\F_{2^m}\simeq \mathrm{GR}(4^m)/(2).$}
 It is known that $\mathrm{GR}(4^m)=\mathcal{T}\oplus 2\mathcal{T}$
($2$-adic decomposition of $\mathrm{GR}(4^m)$),
and that the {\it group of units} of the Galois ring is $\mathrm{GR}(4^m)^*=\mathcal{T^*}\oplus 2\mathcal{T},$ with $\mathcal{T^*}=\mathcal{T}\setminus \{0\}.$
The generalized Frobenius map of $\mathrm{GR}(4^m)$ defined by
$$ f: \mathrm{GR}(4^m)\rightarrow \mathrm{GR}(4^m),\qquad c=a+2b \mapsto c^f=a^2+2b^2$$
is a ring automorphism of $\mathrm{GR}(4^m)$, where $a,b \in \mathcal{T}$. Moreover, if $\sigma$ is a ring automorphism of $\mathrm{GR}(4^m)$, then $\sigma=f^i$ for some $i$, $0\leq i\leq m-1$.
See more details in \cite[Chapter 6]{Wan:4ary}.

%\todo{we do not need Eisenstein integers} The \emph{Eisenstein integers} $\mathbb{E}$ are the complex numbers of the form$$ z = a + b \omega, \qquad\omega = \frac{-1+i\sqrt{3}}2 = e^{2\pi i/3}, \quad a,b\in \mathbb{Z}.$$Given $p\in \mathbb{E} \backslash \{0\}$, we denote by $\mathbb{E}_p$ the ring $\mathbb{E}/p\mathbb{E}$ of residue classes of $\mathbb{E}$ modulo $p$. We are interested in the two cases $\mathbb{E}_2$ and $\mathbb{E}_4$. $\mathbb{E}_2$ is the Galois field $\mathrm{GF}(2^2)$ of characteristic $2$. Its elements are $[0]_2$, $[1]_2$, $[\omega]_2$, and $[\omega]_2$, where $\omega=\omega^2$, and $[x]_p=x+p\mathbb{E}$; but in what follows, we will omit the braces $[\ ]_p$ when naming the residue classes from $\mathbb{E}_p$, $p=2,4$. $\mathbb{E}_4$ is the Galois ring $\mathrm{GR}(4^2)$ of characteristic $4$. Its elements are $2b+a$, $a,b\in \{0, 1, \omega, \omega\}$.

%=============================
\subsection{\large\bf Representation of the Doob graph}\label{ss:doob}
%This subsection is similar to that in \cite{Kro:perfect-doob}. For completeness, we collect it here.

%The set of units $\{ 1, -\omega, \omega, -1, \omega, -\omega \}$ will be denoted by $S$.
%The \emph{Shrikhande graph} $Sh$ is the Cayley graph of the additive group $\mathbb{E}_4^+$ of $\mathbb{E}_4$ with the generating set $S$.

%The ring $\mathbb{E}_4$ itself can be considered as a module of type $\mathbb{Z}_4^2$ over $\mathbb{Z}_4$. Every element $x$ of $\mathbb{E}_4$ can be represented by a pair of coordinates in the basis  $(\omega, 1)$; denote this pair by $\widehat x$.

%will be denoted by $Sh$, too.

Denote by $D(m,n)$ the Cartesian product  $\mathrm{Sh}^m \times K^n$
of $m$ copies
of the Shrikhande graph and $n$ copies
of the complete $4$-vertex graph.
If $m>0$, then $D(m,n)$ is called a Doob graph.
The \emph{Shrikhande graph} $\mathrm{Sh}$
is the Cayley graph of the additive group $\mathbb{Z}_4^{2+}$ of $\mathbb{Z}_4^{2}$
%The Cayley graph of $\mathbb{Z}_4^{2+}$
with the generating set
$S = \{01,30,33,03,10,11\}$.
That is, the vertex
set is the set of elements of $\mathbb{Z}_4^{2}$, two elements being adjacent
if and only if their difference is in $S$.
Next we will use two different representations of
the complete $4$-vertex graph $K = K_4$ as a Cayley graph.
The first representation of $K$ is
the Cayley graph on $\mathbb{Z}_2^{2+}$
with the generating set $\{01,10,11\}$.
At second, $K$ will be considered as
the Cayley graph of $\mathbb{Z}_4^+$
with the generating set $\{1,2,3\}$.

Take the set of $(2m+2n'+n'')$-tuples
$(x_1,{\ldots},x_{2m},y_1,{\ldots},y_{2n'},z_1,$ ${\ldots},z_{n''})$ from
$\mathbb{Z}_4^{2m}\times\mathbb{Z}_2^{2n'}\times\mathbb{Z}_4^{n''}$,
$n'+n''=n$, as the vertex set of $D(m,n)$.
If a code
$C \subset \mathbb{Z}_4^{2m}\times\mathbb{Z}_2^{2n'}\times\mathbb{Z}_4^{n''}$
is closed with respect to addition, then we say it is \emph{additive}.
An additive code is necessarily closed
with respect to multiplication by an element of $\mathbb{Z}_4$.
So, it is in fact a submodule of the module
$\mathbb{Z}_4^{2m}\times\mathbb{Z}_2^{2n'}\times\mathbb{Z}_4^{n''}$
over $\mathbb{Z}_4$.
The natural graph distance in $D(m,n)$
provides a metric on
$\mathbb{Z}_4^{2m}\times\mathbb{Z}_2^{2n'}\times\mathbb{Z}_4^{n''}$,
which will be called the \emph{$D(m,n)$-metric}
(if $m>0$, a \emph{Doob metric}).
The \emph{weight} of a vertex $x$ of $D(m,n)$
is the distance from $x$ to $\overline 0$
(here and in what follows, $\overline 0$ denotes the zero element of the module,
i.e., the all-zero tuple, whose length is clear from the context).

If we study $1$-perfect codes, the vertices of weight $1$ are of special interest.
 Recall that in the case of $\mathbb{Z}_4^{2m}\times\mathbb{Z}_2^{2n'}\times\mathbb{Z}_4^{n''}$ with $D(m,n'+n'')$-metric,
every vertex of weight $1$ has one of the forms
$(0{\ldots}0xy0{\ldots}0|\overline 0|\overline 0)$,
$(\overline 0|0{\ldots}0vw0{\ldots}0|\overline 0)$,
$(\overline 0|\overline 0|0{\ldots}0z0{\ldots}0)$, where $x$ and $v$ are in odd positions,
$xy\in\{01,11,10,03,33,30\}$, $vw\in\{01,11,10\}$, $z\in\{1,2,3\}$,
and the vertical lines separate the three parts of the tuple of length $2m$, $2n'$, and $n''$, respectively.

%=============================
%=============================

%=============================
%=============================
\subsection{\large\bf\boldmath  Additive $1$-perfect codes in Doob graphs}\label{ss:additive}
A $1$-perfect code in a Doob graph $D(m,n)$ is a subset $C$ of $\mathbb{Z}_4^{2m}\times\mathbb{Z}_2^{2n'}\times\mathbb{Z}_4^{n''}$ which is an independent set such that every vertex in $ \mathbb{Z}_4^{2m}\times\mathbb{Z}_2^{2n'}\times\mathbb{Z}_4^{n''} \backslash C$ is adjacent to exactly one vertex in $C$.

\begin{remark}
 In general, the concept of perfect codes
is related with the following bound, known as the sphere-packing bound.
If the distance between any two different elements
of a code in a discrete metric space is more than $2e$,
then the cardinality of the code does not exceed the cardinality of the space divided by the cardinality of a ball of radius $e$.
The codes attending this bound are called perfect, or $e$-perfect.
As was noted in \cite{KoolMun:2000}, nontrivial $e$-perfect codes in the Doob graphs do not exist for $e\ge 2$ (the arguments are based on the known proof of the nonexistence of such codes in the $4$-ary Hamming graphs \cite{Tiet:1973}, \cite{ZL:1973} and on the algebraic connections between the Doob and Hamming graphs).
\end{remark}

  Define $(A \mid A' \mid A'' )$ as a check matrix of a $1$-perfect code $C$ in $D(m,n'+n'')$, that is to say, $C=\{c\in \Z_4^{2m}\times \Z_2^{2n'}\times \Z_4^{n''} \mid (A\mid A'\mid A'' )c^{{\mathrm T}}={\overline 0}^{{\mathrm T}} \}$, with the multiplication
 $(A\mid A'\mid A'' )(z_1\mid z_2\mid z_3)^{{\mathrm T}}$ for $z=(z_1\mid z_2\mid z_3)\in \Z_4^{2m}\times \Z_2^{2n'}\times \Z_4^{n''}$ as $Az_1^{{\mathrm T}}+2\cdot A'z_2^{{\mathrm T}}+A''z_3^{{\mathrm T}}$
 (here, $2\cdot$ can be formally understood as the group homomorphism $0\to 0$, $1\to 2$ from $\Z_2$ to $\Z_4$, acting coordinatewise on the column vector).

For a tuple $z\in \Z_4^{2m}\times \Z_2^{2n'}\times \Z_4^{n''}$,
the value $(A\mid A'\mid A'' ) z^{{\mathrm T}}$ is called the \emph{syndrome} of $z$.
We will say that $s$ is \emph{covered} by the coordinate $i$ with $2(m+n')+1 \leq i \leq 2(m+n')+n''$ or by the pair of coordinates $2i-1$, $2i$ with  $1 \leq i\le m+n'$,
if it is the syndrome of some $e$ of weight $1$
with the only non-zero value in the position $i$ or the only non-zero values in the positions $2i-1$, $2i$, respectively.
We also make an agreement that by a \emph{pair} of coordinates (or of columns of a check matrix) we will always mean a pair of coordinates
of form $2i-1$, $2i$, where $1 \leq i\le m+n'$, i.e., a pair that corresponds to the same $\mathrm{Sh}$ or $K$ component of $D(m,n'+n'')=\mathrm{Sh}^m\times K^{n'}\times K^{n''}$.

The following lemma is a straightforward reformulation of the definition of $1$-perfect codes in terms of check matrices.
\begin{lemma}\label{l:cover}
An additive code in $D(m,n'+n'')$ with a check matrix $(A \mid A' \mid A'' )$ is $1$-perfect if and only if the matrix does not have all-zero columns and
every nonzero syndrome from $\{(A \mid A' \mid A'' )z  \,|\,z\in \Z_4^{2m}\times \Z_2^{2n'}\times \Z_4^{n''}\}$
is uniquely covered by some of the first $m+n'$ pairs of coordinates or some of the last $n''$ coordinates.
\end{lemma}
This fact is a variant of a general principle \cite{HedGuz:2015}
of recognizing additive $1$-perfect codes in abelian groups with different metrics.
If the number of nonzero syndromes equals the number of weight-$1$ words, then it is sufficient to check that every nonzero syndrome is covered at least once.

Let us recall some important results on the additive $1$-perfect codes in Doob graphs.

\begin{lemma}[\cite{Kro:perfect-doob}]\label{1}
Assume that there exists an additive $1$-perfect code in $\Z_4^{2m}\times \Z_2^{2n'}\times \Z_4^{n''}$ with the Doob $D(m,n'+n'')$-metric. Then for some even $\Gamma \geq 0$ and integer $\Delta \geq 2$,
\begin{eqnarray}
  2m+n'+n'' &=& (2^{\Gamma+2\Delta}-1)/3, \label{EQUA1}\\
  3n'+n'' &=& 2^{\Gamma+\Delta}-1,\label{EQUA2} \\
  n'' & \leq & 2^{\Delta}-1,\qquad n''\ne 1.\label{EQUA3}
\end{eqnarray}
\end{lemma}

\begin{lemma}[\cite{Kro:perfect-doob}]\label{2}
For every $m$, $n'$ and $n''$ satisfying the statement of Lemma \ref{1} with even $\Delta$,
there is an additive $1$-perfect code in
$\mathbb{Z}_4^{2m}\times \mathbb{Z}_2^{2n'}\times \mathbb{Z}_4^{n''}$ with $D(m,n'+n'')$-metric.
\end{lemma}

The main result of this paper is the positive solution of the following problem \cite{Kro:perfect-doob}: for every value $(m,n',n'')$ satisfying (\ref{EQUA1}--\ref{EQUA3}) with odd $\Delta$ (except the case (7,0,7), considered in \cite[Sect.~6]{Kro:perfect-doob}), does there exist an additive $1$-perfect code in $\Z_{4}^{2m}\times \Z_2^{2n'}\times \Z_4^{n''}$ with the $D(m,n'+n'')$-metric?

%\hspace{-0.6cm}{\bf Problem} \cite{Kro:perfect-doob}. For every value $(m,n',n'')$ satisfying (\ref{EQUA1}--\ref{EQUA3}) with odd $\Delta$ (except the case (7,0,7), considered in \cite[Sect.~6]{Kro:perfect-doob}), construct an additive $1$-perfect code in $\Z_{4}^{2m}\times \Z_2^{2n'}\times \Z_4^{n''}$ with the $D(m,n'+n'')$-metric or prove the nonexistence of such codes. In particular, does there exist an additive $1$-perfect code in $\Z_{4}^{16}\times \Z_2^{2}\times \Z_4^{4}$ with the $D(8,5)$-metric?

%=============================
%=============================
%=============================
\section{\Large Main construction}\label{s:main}
In this section, to determine the existence of additive $1$-perfect codes in $D(m,n'+n'')$, we firstly list the specific check matrix of a $1$-perfect code in $D(8,1+4)$ corresponding to the case of $\Gamma=0$ and $\Delta=3$.
Secondly, we set about constructing $1$-perfect codes with any odd $\Delta$, $\Gamma=0$ and $n''=4$ in $D(\frac{2^{2\Delta}-2^{\Delta}-8}{6},\frac{2^{\Delta}-5}{3}+4)$ based on the case of
$\Gamma=0$ and $\Delta=3$.
 Next, for any odd $\Delta$ and $\Gamma=0$ we construct $1$-perfect codes in $D(m,n'+n'')$ with $m,~n'$ and $n''$ satisfying conditions (\ref{EQUA1}--\ref{EQUA3}) in Lemma \ref{1}. Finally, infinite $1$-perfect codes are constructed with any odd $\Delta$ and even $\Gamma$ based on the case of $\Delta$ odd and $\Gamma=0$.

%=============================
%=============================
\subsection{\large\bf \boldmath ${n}''=4$, $\Gamma=0$ and $\Delta=3$}\label{ss:403}
  For the case $\Gamma=0$ and $\Delta=3$, there are two values of $(m,n',n'')$ satisfying conditions (\ref{EQUA1}--\ref{EQUA3}) in Lemma \ref{1}. One is $(7,0,7)$, the other is $(8,1,4)$.
 Noting that a $1$-perfect code in $D(7,0+7)$ has been constructed in \cite{Kro:perfect-doob}, we then construct a $1$-perfect code in $D(8,1+4)$.

   We list an example of a $1$-perfect code in $D(8,1+4)$ as follows.
  The check matrix, denoted by $(A_1\mid A_2 \mid A_3)$, is
%   $$\left(
%     \begin{array}{c@{\,}cc@{\,}cc@{\,}cc@{\,}cc@{\,}cc@{\,}cc@{\,}cc@{\,} c|c@{\,}c@{\,}|cccc}
%     1&0 & 2 & 2 & 0 & 1 & 1 & 2 & 2 & 1 & 2 & 3 & 1 & 1 & 3 & 1 & 1&0& 1 & 0 & 0 & 1 \\
%     0&1 & 1 & 0 & 2 & 2 & 2 & 3 & 1 & 2 & 2 & 1 & 0 & 1 & 1 & 2 & 1&1& 0 & 1 & 0 & 1 \\
%     2&2 & 0 & 1 & 1 & 0 & 2 & 1 & 2 & 3 & 1 & 2 & 3 & 2 & 0 & 1 & 0&1& 0 & 0 & 1 & 1
%   \end{array}
%   \right)$$
 $$\left(
    \begin{array}{c@{\,}c@{\,\,\,}c@{\,}c@{\,\,\,}c@{\,}c@{\,\,\,}c@{\,}c@{\,\,\,}c@{\,}c@{\,\,\,}c@{\,}c@{\,\,\,}c@{\,}c@{\,\,\,}c@{\,} c@{\,}|@{\,}c@{\,}c@{\,}|@{\,}cccc}\label{TTT,814}
1&0&2&2&0&1&1&2&2&1&2&3&1&1&1&0&1&0&1&0&0&1\\
0&1&1&0&2&2&2&3&1&2&2&1&2&1&3&1&1&1&0&1&0&1\\
2&2&0&1&1&0&2&1&2&3&1&2&1&2&0&3&0&1&0&0&1&1
\end{array}
  \right).$$
It can be easily checked that every nonzero syndrome is covered.
Indeed, all nonzero syndromes are: $(1,0,0)^{\mathrm T}$, $(1,0,2)^{\mathrm T}$, $(1,2,0)^{\mathrm T}$, $(1,2,2)^{\mathrm T}$, $(1,3,0)^{\mathrm T}$, $(1,2,3)^{\mathrm T}$, $(1,1,0)^{\mathrm T}$, $(1,1,2)^{\mathrm T}$, $(1,1,3)^{\mathrm T}$, $(1,1,1)^{\mathrm T}$, $(2,0,0)^{\mathrm T}$, $(2,2,0)^{\mathrm T}$, $(2,2,2)^{\mathrm T}$,
 and their negatives and cyclic shifts.
The syndromes $(2,2,2)^{\mathrm T}$, $(0,0,2)^{\mathrm T}$, $(1,1,1)^{\mathrm T}$,  $(0,0,1)^{\mathrm T}$ (and their negatives) are covered by the last four coordinates.
The syndrome $(0,2,2)^{\mathrm T}$ and its cyclic shifts are covered by the pair of $\Z_2$ coordinates.
The syndromes $(1,0,2)^{\mathrm T}$, $(0,1,2)^{\mathrm T}$, $(1,1,0)^{\mathrm T}$ (and their negatives) are covered by the $1$st pair of coordinates;
the syndromes $(1,2,2)^{\mathrm T}$, $(2,3,1)^{\mathrm T}$, $(3,1,3)^{\mathrm T}$, by $4$th pair;
the syndrome $(1,1,2)^{\mathrm T}$ and its cyclic shifts, by $7$th;
the syndrome $(0,1,3)^{\mathrm T}$ and its cyclic shifts, by $8$th.
From the matrix, it is easy to see that if some syndrome $(a,b,c)^{\mathrm T}$ is covered, then the cyclic shifts  $(c,a,b)^{\mathrm T}$ and  $(b,c,a)^{\mathrm T}$ are covered too.

  %In detail, we make this check matrix by listing all elements of order $4$.

%=============================
%=============================
 \subsection{\large\bf\boldmath $n''=4$, $\Gamma=0$ and $\Delta$ odd}\label{TTT,40ODD}
 In this subsection, we recursively construct additive $1$-perfect codes in $D(m,n'+n'')$
 for any odd $\Delta\ge 3$, $\Gamma=0$, $m=\frac{2^{2\Delta}-2^{\Delta}-8}{6}$, $n'=\frac{2^{\Delta}-5}{3}$, $n''=4$.
To illustrate the approach, we separately consider the case $\Delta=5$.
%=============================
\subsubsection{\textbf{The first recursive step}}\label{TTT,FIRST}
    Firstly, we start with the case $\Delta=5$, $\Gamma=0$, and we ensure there exists an additive $1$-perfect code in $D(164,9+4)$. To prove the claim, we have to find a check matrix which covers all nonzero syndromes.
        Note that the elements of order $2$ of height $2$ over $\Z_4$ are exactly $(2,0)^{{\mathrm T}}$, $(0,2)^{{\mathrm T}}$, $(2,2)^{{\mathrm T}}$ and the elements of order $4$ of height $2$ over $\Z_4$ are exactly $(0,1)^{{\mathrm T}}$, $(0,3)^{{\mathrm T}}$, $(1,0)^{{\mathrm T}}$, $(1,1)^{{\mathrm T}}$, $(1,2)^{{\mathrm T}}$, $(1,3)^{{\mathrm T}}$,
    $(2,1)^{{\mathrm T}}$, $(2,3)^{{\mathrm T}}$, $(3,0)^{{\mathrm T}}$, $(3,1)^{{\mathrm T}}$, $(3,2)^{{\mathrm T}}$, $(3,3)^{{\mathrm T}}$.
    Note that $164$ pairs exactly cover $164\times 6=8\times 6+3\times 56+12\times 64=8\times 6+3\times | \mathrm{GR}(4^3)^*|+12\times |\mathrm{GR}(4^3)|$ syndromes. To make our construction, we choose any two elements $\mu$, $ \nu$ in $\mathrm{GR}(4^3)^*$ such that $\mu+\nu$ is a unit again, then take $a_1,a_2,\ldots,a_{56} \in \mathrm{GR}(4^3)^*$ with $a_i+a_{57-i}=0$ and $\{a_1,a_2,\ldots,a_{64}\}=\mathrm{GR}(4^3)$.  Then add the matrices
  $$ \left(
     \begin{array}{c@{\ }c@{\ \ .\,.\ \ }c@{\ }cc@{\ }c@{\ \ .\,.\ \ }c@{\ }cc@{\ }c@{\ \ .\,.\ \ }c@{\ }c}
       a_1\mu & a_1\nu  & a_{28}\mu & a_{28}\nu & a_1\mu & a_1\nu &   a_{64}\mu & a_{64}\nu &
         a_1\mu & a_1\nu & a_{64}\mu & a_{64}\nu\\
       2 & 0   & 2 & 0 & 1 & 0  & 1 & 0 & 1 & 1 &   1 & 1 \\
       0 & 2 & 0 & 2 & 0 & 1 &   0 & 1 & 2 & 3 &   2 & 3 \\
     \end{array}
   \right),$$
  $$\frac{1}{2}\left(
                          \begin{array}{c@{\ }c@{\ \ .\,.\ \ } c@{\ }c@{\ } }
                            a_{57}\mu & a_{57}\nu  & a_{64}\mu & a_{64}\nu \\
                            2 & 0  & 2 & 0 \\
                            0 & 2  & 0 & 2 \\
                          \end{array}
                        \right)$$
                         to the left ($\Z_4^2$-part) and the middle ($\Z_2^2$-part) parts of the matrix $
   \left(
     \begin{array}{c|c|c}
       A_1 & A_2 & A_3 \\
       \overline 0 & \overline 0 & \overline 0 \\
       \overline 0 & \overline 0 & \overline 0 \\
     \end{array}
   \right)
   ,$ respectively, where $(A_1 \mid A_2 \mid A_3)$ is the check matrix of a $1$-perfect code in $D(8,1+4)$ constructed in Section~\ref{ss:403}.
    Then it is easy to check that all coordinates of the new check matrix, which could be seen as the combination of $
    \Z_4^2$-part, $\Z_2^2$-part and $\Z_4$-part,
   cover all nonzero syndromes. In fact, for any distinct $i,j=1,2,\ldots,64$, we have $a_i\mu \neq a_j\mu$, $a_i\nu \neq a_j\nu$ and $a_i(\mu+\nu) \neq a_j(\mu+\nu)$.
   The number of syndromes of order $2$ is $2^5-1=31$ while the  $\Z_2^2$-part covers $3\times(8+1)=27$ syndromes and the $\Z_4$-part covers $4$ syndromes of order $2$. And the number of syndromes of order $4$ is $(2^5-1)2^5=2^{10}-2^5=992$, while the corresponding coordinates of the $\Z_4^2$-part and the $\Z_4$-part cover $164\times 6+4\times 2=992$ syndromes.
%=============================
\subsubsection{\textbf{The general case}}\label{ss:delta-gen}
  For the case $\Gamma=0, ~(\Delta-2)$ odd, let $(H\mid H'\mid H'')$ be a check matrix of a $1$-perfect code in $D(\widetilde{m},\widetilde{n'}+4)$, where $\widetilde{m}= \frac{2^{2(\Delta-2)}-2^{\Delta-2}-8}{6}$ and $\widetilde{n'}=\frac{2^{\Delta-2}-5}{3}$.  For the case $\Gamma=0$ and $\Delta$ odd, to obtain a $1$-perfect code in $D(m,n'+n'')$ we construct the check matrix as follows by noting that $3(n'-\widetilde{n'})=3\times 2^{\Delta-2}$ and $6m-6\widetilde{m}=15\times 4^{\Delta-2}-3\times 2^{\Delta-2}=12\times 4^{\Delta-2}+3( 4^{\Delta-2}- 2^{\Delta-2})$.

  We identify the elements of $\mathrm{GR}(4^{\Delta-2})$ with the corresponding vectors of length $\Delta-2$ over $\Z_4$, as described in Section~\ref{s:prel}.
  Choose any two distinct elements  $\alpha$, $\beta$ in $\mathrm{GR}(4^{\Delta-2})^{*}$ such that $\alpha+\beta$ is a unit again.
  %Let $\mathcal{T}$ be the set $\{x\in GR(4^{\Delta-2})|x^{2^{\Delta-2}=x}\}$. Consequently, there exists an element $\xi \in GR(4^{\Delta-2})$ such that $\mathcal{T}=\{0,1,\xi,\xi^2,\ldots,\xi^{2^{\Delta-2}-2}\}$.
  Denote by $c_1,c_2,\ldots,c_{t}$ the units of $\mathrm{GR}(4^{\Delta-2})$.
  %For the set $2\mathcal{T}$ we denote it by $c_{t+1},c_{t+2},\ldots,c_{s}$.
   And denote by $c_1,c_2,\ldots,c_{s}$ the all elements of $\mathrm{GR}(4^{\Delta-2})$, where $t=(2^{\Delta-2}-1)2^{\Delta-2}$, $c_i+c_{t+1-i}=0$ for all $i=1,\ldots,t$ and $s=4^{\Delta-2}$.
   Define the matrices $B$, $W$, $V$, $D'$, $E$, $E'$, $E''$ as follows:
$$B=\left(
                                \begin{array}{ccccccc}
                                  c_1\alpha & c_1\beta & c_2\alpha & c_2\beta & \cdots & c_{\frac{t}{2}}\alpha & c_{\frac{t}{2}}\beta \\
                                  2 & 0 & 2 &0 & \cdots & 2 & 0 \\
                                  0 & 2 & 0& 2 & \cdots & 0 & 2 \\
                                \end{array}
                              \right);$$
$$W=\left(
                                 \begin{array}{c@{\ }cc@{\ }ccc@{\ }c}
                                   c_1\alpha & c_1\beta & c_2\alpha & c_2\beta & \cdots & c_{s}\alpha & c_{s}\beta \\
                                     0 & 1 & 0 & 1 & \cdots & 0 & 1   \\
                                     1 & 0 & 1 & 0 & \cdots & 1 & 0  \\
                                 \end{array}
                               \right),
\quad V=
\left(
                                 \begin{array}{c@{\ }cc@{\ }ccc@{\ }c}
                                   c_1\alpha & c_1\beta & c_2\alpha & c_2\beta & \cdots & c_{s}\alpha & c_{s}\beta  \\
                                     1 & 1 & 1& 1 & \cdots & 1 & 1  \\
                                    2 & 3 & 2 &3 & \cdots & 2 & 3  \\
                                 \end{array}
                               \right);$$
  $$D'=\frac12\left(
                          \begin{array}{c@{\ }cc@{\ }ccc@{\ }c}
                            c_{t+1}\alpha & c_{t+1}\beta & c_{t+2}\alpha & c_{t+2}\beta & \cdots & c_s\alpha & c_s\beta \\
                            2 & 0 & 2 & 0 & \cdots & 2 & 0 \\
                            0 & 2 & 0 & 2 & \cdots & 0 & 2 \\
                          \end{array}
                        \right);$$
$$
E=\left(
                                  \begin{array}{c}
                                    H \\
                                    \mathbf{0} \\
                                    \mathbf{0} \\
                                  \end{array}
                                \right),
\qquad
E'=\left(
                                  \begin{array}{c}
                                    H' \\
                                    \mathbf{0} \\
                                    \mathbf{0} \\
                                  \end{array}
                                \right),
\qquad
E''=\left(
                                  \begin{array}{c}
                                    H'' \\
                                    \mathbf{0} \\
                                    \mathbf{0} \\
                                  \end{array}
                                \right).
$$
Denote by $M$ the matrix $(BWVE\mid E' D'\mid E'')$.
Keeping the notations above, we have the following proposition.
\begin{proposition}\label{111}
Let the code $C$ be defined by the check matrix  $M$, constructed as above. Then $C$ is a $1$-perfect code in the Doob graph $D(m,n'+n'')$, where $m=\frac{2^{2\Delta}-2^{\Delta}-8}{6}$, $n'=\frac{2^{\Delta}-5}{3}$, $n''=4$ and $\Delta$ is odd.
\end{proposition}
\begin{proof}
Note that there are at most $16^\Delta $
different syndromes.
Let us consider an arbitrary
$z\in \mathbb{Z}_4^{2m}\times \mathbb{Z}_2^{2n'}\times \Z_4^{n''}$
and its syndrome $s=Mz^{{\mathrm T}}$.
If $s$ is the all-zero column,
then $z\in C$.
Let us show that if $s$ is non-zero,
then there is a unique codeword $c=z-e$ adjacent to $z$.
For the existence, it is sufficient to find a weight-$1$ tuple $e$
with syndrome $s$.
Let us consider two cases.
\begin{enumerate}
  \item[(i)] If the order of $s$ is $2$, then it is covered by the $\Z_2^2$-part and the $\Z_4$-part.
Indeed, there are $2^\Delta$ elements of order $2$ in $\mathrm{GR}(4^\Delta)$,
while $3\widetilde{n'}=2^{\Delta-2}-5$ distinct syndromes with the last two rows $(0,0)^{\mathrm T}$ are covered by corresponding coordinates of $E'$ and $3(\widetilde{n'}-\widetilde{n''})=3\times 2^{\Delta-2}$ distinct syndromes with the last two rows $(2,0)^{\mathrm T}$, $(0,2)^{\mathrm T}$, $(2,2)^{\mathrm T}$ are covered by corresponding coordinates of $D'$. Except that, $E''$ covers $4$ distinct syndromes with the last two rows $(0,0)^{\mathrm T}$. Totally, $M$ covers $2^{\Delta-2}-5+3\times 2^{\Delta-2}+4=2^\Delta-1$ distinct elements in $\Z_4^\Delta$. That is to say, all syndromes of order $2$ are covered by corresponding coordinates of $M$.
  \item[(ii)] If the order of $s$ is $4$, then  it is covered by corresponding coordinates of the first part and corresponding coordinates of the third part. Indeed, there are $2^{2\Delta}-2^\Delta$ elements of order $4$ in $\mathrm{GR}(4^\Delta)$, while $6\widetilde{m}=2^{2(\Delta-2)}-2^{\Delta-2}-8$ distinct syndromes with the last two rows $(0,0)^{\mathrm T}$ are covered by $E$,
 $6\times \frac{t}{2}=3(2^{\Delta-2}-1)2^{\Delta-2}$ distinct syndromes with the last two rows $(2,0)^{\mathrm T}$, $(0,2)^{\mathrm T}$, $(2,2)^{\mathrm T}$ are covered by $B$, and $6\times s\times 2=12\times 4^{\Delta-2}$ distinct syndromes with the last two rows $(0,1)^{{\mathrm T}}$, $(0,3)^{{\mathrm T}}$, $(1,0)^{{\mathrm T}}$, $(1,1)^{{\mathrm T}}$, $(1,2)^{{\mathrm T}}$, $(1,3)^{{\mathrm T}}$,
  $ (2,1)^{{\mathrm T}}$, $(2,3)^{{\mathrm T}}$, $(3,0)^{{\mathrm T}}$,
    $(3,1)^{{\mathrm T}}$, $(3,2)^{{\mathrm T}}$, $(3,3)^{{\mathrm T}}$ are covered by corresponding coordinates of $W$ and $V$. Except that,
       $2\times 4=8$ distinct syndromes with the last two rows $(0,0)^{\mathrm T}$ are covered by corresponding coordinates of $E''$. Totally, corresponding coordinates of $M$ covers $2^{2(\Delta-2)}-2^{\Delta-2}-8+3(2^{\Delta-2}-1)2^{\Delta-2}+12\times 4^{\Delta-2}+8=2^{2\Delta}-2^{\Delta}$ distinct elements of order $4$ in $\Z_4^\Delta$. That is to say, all syndromes of order $4$ are covered by corresponding coordinates of $M$.
\end{enumerate}

It is easy to see that the choice of $e$ is unique.
\qed\end{proof}

%=============================
%=============================
\subsection{\large\bf\boldmath Increasing $n''$ when $\Gamma=0$ and $\Delta$ odd}\label{ss:n''}
   To construct more $1$-perfect codes we want to increase $n''$ based on the above check matrix.
   We start with a simple case and end up with a generalized case in this subsection.
%=============================
\subsubsection{\textbf{\boldmath The special case $(7,0+7)$}}\label{sss:707}
An additive $1$-perfect code has already been found in \cite{Kro:perfect-doob}; we recall its description in Section~\ref{TTT,707}. However, to illustrate the technique of increasing $n''$, we construct another code. In Section~\ref{TTT,sub}, we prove that this code is not equivalent to that of \cite{Kro:perfect-doob}.

We begin with the $1$-perfect code in $D(8,1+4)$, see Subsection \ref{TTT,814}. Note that the part over $\Z_2$ covers  three syndromes $(2, 2, 0)^{{\mathrm T}}$,  $(0, 2, 2)^{{\mathrm T}}$, $(2, 0, 2)^{{\mathrm T}}$, which can be written as $2(3, 1, 0)^{{\mathrm T}},$ $2(3+1, 1+2, 0+1)^{{\mathrm T}}$, $2(1, 2, 1)^{{\mathrm T}}$. At the same time, the last two columns of the first part over $\Z_4$ cover six syndromes $(3,1,0)^{{\mathrm T}}$, $(1,2,1)^{{\mathrm T}}$, $3(3,1,0)^{{\mathrm T}}=(1,3,0)^{{\mathrm T}}$, $3(1,2,1)^{{\mathrm T}}=(3,2,3)^{{\mathrm T}}$, $(3+1,1+2,0+1)^{{\mathrm T}}=(0,3,1)^{{\mathrm T}}$,
 $3(3+1,1+2,0+1)^{{\mathrm T}}=(0,1,3)^{{\mathrm T}}$. Note that these nine syndromes are exactly $k(3,1,0)^{{\mathrm T}} $, $k(1,2,1)^{{\mathrm T}}$, $k(0,3,1)^{{\mathrm T}}$ for $k=1,2,3$. That means it is feasible to increase $n''$ by adding three columns $(3,1,0)^{{\mathrm T}}$, $(1,2,1)^{{\mathrm T}}$, $ (0,3,1)^{{\mathrm T}}$ and deleting the last two columns of the first part and the two columns of the second part over $\Z_2$.

%=============================
\subsubsection{\textbf{The general case}}\label{sss:n''gen}
Based on the codes constructed in Subsection \ref{ss:delta-gen}, we start from the check matrix $M$.

Generally, the corresponding coordinates of every pair of columns $B_1$ and $ B_2$ over $\Z_4$ from the first part of the matrix cover six syndromes $B_1$, $B_2$, $B_1+B_2$, $3B_1$, $3B_2$, $3(B_1+B_2)$.

If $D_1$, $D_2$ is a pair of  columns in $D'$, then $(2D_1,2D_2)$ is of the form $\left(
                                                        \begin{array}{cc}
                                                          2c_i\alpha & 2c_i\beta \\
                                                          0 & 2 \\
                                                          2 & 0 \\
                                                        \end{array}
                                                      \right)$
for some $i\in\{1,2,\ldots,s\}$ (the choice of $i$ is not unique in general).
By the definition of $W$, it contains the pair of columns $(B_1,B_2)=\left(
                                                        \begin{array}{cc}
                                                          c_i\alpha & c_i\beta \\
                                                          0 & 1 \\
                                                          1 & 0 \\
                                                        \end{array}
                                                      \right)$.
This pair covers the syndromes
 $B_1$, $B_2$, $B_1+B_2$,
 $3B_1$, $3B_2$, $3B_1+3B_2$,
 while the syndromes
  $2B_1$, $2B_2$, $2B_1+2B_2$,
are covered by $(D_1,D_2)$.
That implies we can construct an additive $1$-perfect code in $D(m-1,(n'-1)+(4+3))$ by deleting these two pairs but adding the three columns
$(B_1,B_2,B_1+B_2)=\left(
 \begin{array}{ccc}
  c_i\alpha & c_i\beta & c_i(\alpha+\beta) \\
   1 & 0 & 1 \\
   0 & 1 & 1 \\
   \end{array}
   \right)$
to the third part of the matrix.

 Remembering that the matrix $(E\mid E'\mid E'')$ was obtained at the previous recursive step or corresponds to the case $(8,1+4)$ we can apply the same strategy as above or Subsection \ref{sss:707}. So, for every pair of columns $D_1$, $D_2$ in $D'$ or $E'$, we can find a pair $B_1$, $B_2$ in $W$ or $E$ such that $2B_1=2D_1$ and $2B_2=2D_2$. Then we can replace these $4$ columns by the new columns $B_1$, $B_2$, $B_1+B_2$ in the third part of the matrix.
  %Generally, the corresponding coordinates of every pair of columns $B_1$, $ B_2$ over $\Z_4$ from the part of the matrix covers six syndromes $B_1$, $B_2$, $B_1+B_2$, $3B_1$, $3B_2$, $3(B_1+B_2)$. It is not difficult to find %\todo{not clear. We must refer the concrete construction above, choose the first pair from the submatrix $W'$ and show how to find the corresponding pair in the submatrix $D$}that $2B_1$, $2B_2$, $2(B_1+B_2)$ are covered by the pair of two columns over $\Z_2$. Then we could delete these nine columns and add $B_1$, $B_2$, $B_1+B_2$ to the third part over $\Z_4.$
   Using that algorithm, we can increase $n''$ up to $2^{\Delta}-1$ and decrease $n'$ down to $0$. Let $\overline{M}$ be the new matrix  constructed as above.
    %\todo{rephrase; bad language}.

 So, once we have a $1$-perfect code in $D(m,n'+4)$ constructed as in Subsection~\ref{ss:delta-gen}, we also have additive $1$-perfect codes in $D(m-1,(n'-1)+(4+3))$, $D(m-2,(n'-2)+(4+6))$, $\ldots$,  $D(m-n',0+(2^{\Delta}-1))$.
Keeping the notations above, We obtain the following statement.

\begin{proposition}\label{TTT,PROP}
%Let $\overline{M}$ be the check matrix of a code $C$. Consequently,
Let the code $C$ be defined by the check matrix  $\overline{M}$, constructed as above. Then
$C$ is a $1$-perfect code in Doob graphs $D(m,n'+n'')$, where $m$, $n'$, $n''$ satisfy conditions (\ref{EQUA1}--\ref{EQUA3}) in Lemma \ref{1} with $\Gamma=0$ and $\Delta$ odd.
\end{proposition}
%\begin{proof}
%It naturally follows by Proposition \ref{111}.
%\qed\end{proof}
%=============================
%=============================
\subsection{\large\bf\boldmath Arbitrary even $\Gamma$ and odd $\Delta$}\label{TTT,EVEN}

  In this subsection, we are aimed at constructing a check matrix of a $1$-perfect code in $D(m^*,n'^*+n'')$, from a check matrix $(I|I'|I'')$ of a $1$-perfect code in $D(m,n'+n'')$ without rows of order $2$, where $6m^*-6m=(2^{\Gamma}-1)(2^{2\Delta}-2^{\Delta})$ and $3n'^*-3n'=(2^{\Gamma}-1)2^{\Delta}$.

The idea is the same as when we increased $\Delta$ in Section~\ref{TTT,40ODD},
but instead of acting recursively,
we increase the number of order-$2$ rows from $0$ to $\Gamma$ is one step
(the reason is that $\Z_4^{\Delta}\times (2\Z_2)^{\Gamma-2}$ cannot be
represented as a Galois ring for $\Gamma>2$).

Let the triples $\{\mathbf{a}_i,\mathbf{b}_i,\mathbf{c}_i\}$,
 $i=1,\ldots,{\frac{2^{\Gamma}-1}{3}}$,
 such that $\mathbf{a}_i^{\mathrm T},\mathbf{b}_i^{\mathrm T},\mathbf{c}_i^{\mathrm T} \in 2\mathbb{Z}_2^{\Gamma} \setminus \{ {\overline 0} \}$
and
$\mathbf{a}_i^{\mathrm T}+\mathbf{b}_i^{\mathrm T}+\mathbf{c}_i^{\mathrm T}=\overline 0$
form a partition of $2\mathbb{Z}_2^{\Gamma} \setminus \{ {\overline 0}\}$,
i.e., $$\bigcup_{i=1}^{({2^{\Gamma}-1})/{3}}\{\mathbf{a}_i^{\mathrm T},\mathbf{b}_i^{\mathrm T},\mathbf{c}_i^{\mathrm T}\}=2\mathbb{Z}_2^{\Gamma} \setminus \{ {\overline 0}\} $$
(such partition can be easily constructed from the multiplicative cosets of the subfield $\mathrm{GF}(4)$ in the field $\mathrm{GF}(2^\Gamma)$).
%In fact, $\F_{2^2}$ is a subfield of $\F_{2^{\Gamma}}$, we have that $\F_{2^2} \backslash \{0\} $ is a multiplicative subgroup of $\F_{2^{\Gamma}} \backslash \{0\}$. By the coset factorization we know all $S_i$ could form a partition of $\F_2^{\Gamma} \backslash\{0\}$.
Let $u_1$, $u_2$, \ldots, $u_l$ be the units of $\mathrm{GR}(4^{\Delta})$ and $u_{l+1}$, \ldots, $u_{k}$ the non-units in $\mathrm{GR}(4^{\Delta})$, where $u_j+u_{l+1-j}=0$ for $j=1,2,\ldots, l$ and $l=(2^{\Delta}-1)2^{\Delta}$, $k=4^{\Delta}$. Choose any two elements $\gamma$, $\delta$ in $\mathrm{GR}(4^{\Delta})^*$ such that $\gamma+\delta$ is also a unit.
 Define the matrices $F=(F_1,\ldots,F_{\frac{2^{\Gamma}-1}{3}})$, $G'=(G'_1,\ldots,G'_{\frac{2^{\Gamma}-1}{3}})$, $E$, $E'$, $E''$:
 %\todo{the indices should not be in bold, anywhere in the text}
 $$F_i=\left(
                            \begin{array}{c@{\ }cc@{\ }ccc@{\ }c}
                              u_1\gamma & u_1\delta & u_2\gamma & u_2\delta & \cdots & u_{\frac{l}{2}}\gamma & u_{\frac{l}{2}}\delta \\
                              \mathbf{a}_i & \mathbf{b}_i &  \mathbf{a}_i & \mathbf{b}_i & \cdots & \mathbf{a}_i & \mathbf{b}_i
                            \end{array}
                          \right),\qquad i=1,\ldots,\frac{2^{\Gamma}-1}{3};$$
$$G'_i= \frac12\left(
                            \begin{array}{c@{\ }cc@{\ }ccc@{\ }c}
                              u_{l+1}\gamma & u_{l+1}\delta & u_{l+2}\gamma & u_{l+2}\delta & \cdots & u_{k}\gamma & u_{k}\delta \\
                              \mathbf{a}_i & \mathbf{b}_i &  \mathbf{a}_i & \mathbf{b}_i & \cdots & \mathbf{a}_i & \mathbf{b}_i
                            \end{array}
                          \right),\qquad i=1,\ldots,\frac{2^{\Gamma}-1}{3};$$
$$J=\left(
                           \begin{array}{c}
                             I\\
                             \mathbf{0} \\
                           \end{array}
                         \right),
\qquad J'=\left(
                           \begin{array}{c}
                             I' \\
                             \mathbf{0} \\
                           \end{array}
                         \right),
\qquad J''=\left(
                           \begin{array}{c}
                             I'' \\
                             \mathbf{0} \\
                           \end{array}
                         \right).
$$
 Then, we denote by  $\widehat M$ the matrix $(FJ\mid J'G'\mid J'')$.

 \begin{theorem}\label{TTT,THE}
 Let $\Gamma$ be even and $\Delta$ be odd, and let the matrix $\widehat M$ be constructed as above. The set $C=\{c\in \Z_4^{2{m}^*}\times \Z_2^{2{n'}^*}\times \Z_4^{{n''}}\mid \widehat Mc^{{\mathrm T}}={\overline 0}^{{\mathrm T}}\}$ is an additive $1$-perfect code in the Doob graph $D(m,n'+n'')$.
 \end{theorem}
\begin{proof}
Similarly to the proof of Proposition~\ref{111}, we
assume that the syndrome has the form of $\left(
                      \begin{array}{c}
                        \epsilon \\
                        \varepsilon \\
                      \end{array}
                    \right)
$  with $\epsilon^{\mathrm T} \in \Z_4^\Delta$ and $\varepsilon^{\mathrm T} \in 2\Z_2^\Gamma$.

Since $(I|I'|I'')$ is a check matrix of a $1$-perfect code,
the case $\varepsilon = \overline 0^{\mathrm T}$ is covered by the columns of $J$, $J'$, $J''$.

If $\varepsilon$ is nonzero,
then it is uniquely represented as
$\mathbf{a}_i$, $\mathbf{b}_i$,
or $\mathbf{c}_i$ for some $i$ from
$1$ to $\frac{2^{\Gamma}-1}{3}$.
Depending on $\varepsilon=\mathbf{a}_i$, $\varepsilon=\mathbf{b}_i$,
or $\varepsilon=\mathbf{c}_i$, we divide
$\epsilon$ by $\gamma$, $\delta$, or $\gamma+\delta$,
and obtain $u_j$ for some $j$ from $1$ to $4^\Delta$.
So, the syndrome has the form
$\displaystyle\binom{u_j\gamma}{\mathbf{a}_i}$,
$\displaystyle\binom{u_j\delta}{\mathbf{b}_i}$, or
$\displaystyle\binom{u_j\gamma+u_j\delta}{\mathbf{a}_i+\mathbf{b}_i}$.
 If $j\le \frac{l}2$, then the syndrome is covered by the pair of columns
$\left(\begin{array}{cc} u_j\gamma & u_j\delta \\ \mathbf{a}_i & \mathbf{b}_i \end{array}\right)$ of $F_i$.
 If $\frac{l}2 < j\le l$, then the syndrome is covered by the pair of columns
$\left(\begin{array}{cc} u_{l+1-j}\gamma & u_{l+1-j}\delta \\ \mathbf{a}_i & \mathbf{b}_i \end{array}\right)$ of $F_i$.
 If $j>l$, then the syndrome is covered by the pair of columns
$\left(\begin{array}{cc} \frac12 u_j\gamma & \frac12 u_j\delta \\ \frac12\mathbf{a}_i & \frac12\mathbf{b}_i \end{array}\right)$ of $G'_i$.
By numerical reasons, every syndrome is covered exactly once.
Thus, the proof is completed.
\qed\end{proof}

\begin{corollary}\label{TTT,CORO}
For every $m$, $n'$ and $n''$ satisfying the statement of Lemma~\ref{1} with odd $\Delta$,
there is a $1$-perfect code in
$\mathbb{Z}_4^{2m}\times \mathbb{Z}_2^{2n'}\times \mathbb{Z}_4^{n''}$ with $D(m,n'+n'')$-metric.

\end{corollary}

Combining with Lemmas~\ref{1} and~\ref{2}, we get necessary and sufficient conditions of the existence of additive $1$-perfect codes in $D(m,n'+n'')$.

\begin{theorem}\label{TTT,2}
Additive $1$-perfect codes in
$\mathbb{Z}_4^{2m}\times \mathbb{Z}_2^{2n'}\times \mathbb{Z}_4^{n''}$ with $D(m,n'+n'')$-metric exist if and only if
$m$, $n'$ and $n''$ satisfy (\ref{EQUA1}--\ref{EQUA3}) for some nonnegative integer $\Gamma$ and $\Delta$.
\end{theorem}
In addition, we note that (\ref{EQUA1}--\ref{EQUA3}) imply  $\Gamma$ is even and $\Delta\ne 1$. Moreover, $\Delta=0$ implies that $m=n''=0$; in this case $D(m,n'+n'')$ is a Hamming graph,  not a Doob graph.

%=============================
%=============================
%=============================
\section{\Large Quasi-cyclic $1$-perfect codes}\label{TTT,QUASI}
Two codes $\mathcal{C}_1$ and $\mathcal{C}_2$ in a graph are \emph{equivalent} if there is an automorphism of the graph that sends $\mathcal{C}_1$ to $\mathcal{C}_2$.

In this section, we list three quasi-cyclic $1$-perfect codes.
For each of these codes, we describe a check matrix
whose columns are defined in terms of a primitive root $\xi\in \mathrm(GR)(4^\Delta)$ of an irreducible polynomial of order $\Delta$, ($\Delta=3$, $5$, $7$) over $\Z_4$.
In each case, multiplication of the columns by $\xi$ is equivalent to a coordinate permutation consisting of $(2m+n'')/(2^\Delta-1)$ cycles of order $2^\Delta-1$. It follows that such permutation stabilizes the code, and the code is quasi-cyclic.
In the end of this section, we prove that  each of these three codes is not equivalent to any of the codes constructed in Section~\ref{s:main}.

%=============================
%=============================
\subsection{\large\bf\boldmath The $1$-perfect code in $D(7,0+7)$ (the case $\Gamma=0$, $\Delta=3$)}\label{TTT,707}

 Let $\xi \in \mathrm{GR}(4^{3})$ be a primitive root of the basic irreducible polynomial $x^3+2x^2+x+3$ over $\Z_4$. The check matrix of the  quasi-cyclic  $1$-perfect code in $D(7,0+7)$ constructed  in \cite{Kro:perfect-doob}
  consists of the pairs of columns $\xi^i+ 2\xi^{i+2}$, $\xi^{i+1}+ 2\xi^{i+5}$ in the left part and
the columns $\xi^i$ in the right part, $i=0,1,2,3,4,5,6$.

%=============================
%=============================
\subsection{\large\bf\boldmath A $1$-perfect code in $D(155,0+31)$ (the case $\Gamma=0$, $\Delta=5$)}\label{TTT,05}

%\begin{enumerate}
  %\item[(a)]
   % Let $\xi \in GR(4^{5})$ be such that $GR(4^5)=\mathcal{T}\oplus 2\mathcal{T}$ with   $\mathcal{T}=\{0,1,\xi,\xi^2,\ldots,\xi^{2^{5}-2}\}$. Let $1$, $\xi$, $\xi^2$, \ldots, $\xi^{2^{5}-2}$ be the $31$ columns of $H$ in the last part of coordinates.

  %\item[(b)]
%Let $\xi^{i+1}+2\xi^{i+2}$ and $\xi^{i+2}+2\xi^{i+5}$ be a pair of two columns of $H$ over $\Z_4$. By calculating we get$(\xi^{i+1}+2\xi^{i+2})+(\xi^{i+2}+2\xi^{i+5})=\xi^{i+19}+2\xi^{i+30}$, $-(\xi^{i+1}+2\xi^{i+2})=\xi^{i+1}+2\xi^{i+19}$, $-(\xi^{i+2}+2\xi^{i+5})=\xi^{i+2}+2\xi^{i+31}$ and $-(\xi^{i+19}+2\xi^{i+30})=\xi^{i+19}+2\xi^{i+38}$.Let $i=0,1,2,\cdots,30$, then we will get $31$ pairs of two columns of $H$ over $\Z_4$.\item[(c)] We list cyclotomic coset   modulo $31$  as follows.$S_1=\{1,2,3,8,16\}$, $S_2=\{3,6,12,24,17\}$, $S_3=\{5,10,20,9,18\}$, $S_4=\{7,14,28,\\25,19\}$, $S_{5}=\{11,22,13,26,21\}$, $S_{6}=\{15,30,29,27,23\}$.Note that the difference of power of two terms of $6$ elements in $(b)$ are exactly $1\in S_1$, $3\in S_2$, $18\in S_3$, $19\in S_{4}$, $11\in S_{5}$, $29\in S_{6}$.Let $\xi^{l(i+1)}+2\xi^{l(i+2)}$ and $\xi^{l(i+2)}+2\xi^{l(i+5)}$ be another pair of two columns of $H$ over $\Z_4$, where $l=2,4,8,16.$

%  \item Let $H$  consists of $155$ pair of columns in the form  $\xi^{l(i+1)}+2\xi^{l(i+2)}$ and $\xi^{l(i+2)}+2\xi^{l(i+5)}$ with $l=1,2,4,8,16$, $i=0,1,2,\cdots,30$ and $31$ columns in the form of $\xi^i$ with $i=0,1,2,\cdots,30.$\end{enumerate}

\begin{proposition}\label{p:155,31}
Let $\xi$ be a primitive root of the basic irreducible polynomial $x^5+3x^2+2x+3$ over $\Z_4$.
Let $H$ be the $5\times 341$ matrix over $Z_4$ consisting of $155=31\cdot 5$ pairs of columns $\xi^{2^l(i+1)}+2\xi^{2^l(i+2)}$, $\xi^{2^l(i+2)}+2\xi^{2^l(i+5)}$ with $l=1,2,3,4,5$, $i=0,1,2,\ldots,30$ in the left part and $31$ columns  $\xi^i$, $i=0,1,2,\ldots,30,$ in the right part.
The code $C$ defined by the check matrix $H$ is a $1$-perfect code in $D(155,0+31)$.
\end{proposition}
\begin{proof}
To check whether $C$ is a $1$-perfect code in $D(155,0+31)$, we need to verify that all syndromes in $\Z_4^5$ are covered by coordinates of $H$. Identify the elements of $\Z_4^5$ with the elements in $\mathrm{GR}(4^5)$.
Since $\xi$ is a primitive root of the polynomial $x^5+3x^2+2x+3$,  we have $\mathrm{GR}(4^5)=\mathcal{T}\oplus 2\mathcal{T}$ with   $\mathcal{T}=\{0,1,\xi,\xi^2,\ldots,\xi^{2^{5}-2}\}$. It is sufficient to show that $\xi^i, 2\xi^i, \xi^i+2\xi^j$ with $i,j\in\{0,1,2,\ldots,30\}$ are covered by coordinates of $H$.
 %Let $1$, $\xi$, $\xi^2$, \ldots, $\xi^{2^{5}-2}$ be the $31$ columns of $H$ in the last part of coordinates.

Firstly, put $c=\xi^{i+1}+2\xi^{i+2}$ and $c'=\xi^{i+2}+2\xi^{i+5}$. By calculating, we get
\begin{eqnarray*}
&-c=\xi^{i+1}+2\xi^{i+19},\qquad -c'=\xi^{i+2}+2\xi^{i+31},
\\
& c+c'=\xi^{i+19}+2\xi^{i+30},\qquad  -(c+c')=\xi^{i+19}+2\xi^{i+38}.
\end{eqnarray*}

%Let $i=0,1,2,\cdots,30$, then we will get $31$ pairs of two columns of $H$ over $\Z_4$.
Then, we list cyclotomic cosets $2x$  modulo $31$  with $x=1,2,\ldots,30$ as follows:
$$\begin{array}{lll}
S_1=\{1,2,3,8,16\},& S_2=\{3,6,12,24,17\},& S_3=\{5,10,20,9,18\},\\
S_4=\{7,14,28,25,19\},& S_{5}=\{11,22,13,26,21\},& S_{6}=\{15,30,29,27,23\}.
  \end{array}$$
  Note that the difference $b-a$ of power of two terms ($\xi^a$ and $2\xi^b$) of the $6$ elements $\pm c$, $\pm c'$, $\pm(c+c')$ are exactly $1\in S_1$, $3\in S_2$, $18\in S_3$, $19\in S_{4}$, $11\in S_{5}$, $29\in S_{6}$.

  Let $f$ be the automorphism of $\mathrm{GR}(4^5)$ defined in Subsection \ref{ss:galois}, then $f^l$ is also automorphism.
  Let $f^l(c)=\xi^{2^l(i+1)}+2\xi^{2^l(i+2)}$ and
  $f^l(c')=\xi^{2^l(i+2)}+2\xi^{2^l(i+5)}$
   be pairs of $H$ over $\Z_4$, where $l=0,1,2,3,4.$ Since $f^l$ is a homomorphism, we have
   \begin{eqnarray*}
    f^l(c)+f^l(c')=f^l(c+c')&=&\xi^{2^l(i+19)}+2\xi^{2^l(i+30)},\\
-f^l(c)=f^l(-c)&=&\xi^{2^l(i+1)}+2\xi^{2^l(i+19)},\\
-f^l(c')=f^l(-c')&=&\xi^{2^l(i+2)}+2\xi^{2^l(i+31)},\\
-(f^l(c)+f^l(c'))=f^l(-c)+f^l(-c')=f^l(-(c+c'))&=&\xi^{2^l(i+19)}+2\xi^{2^l(i+38)}.
   \end{eqnarray*}

It is easy to see that $2^l\cdot 1\in S_1$, $2^l\cdot 3\in S_2$, $2^l\cdot 18\in S_3$, $2^l\cdot 19\in S_{4}$, $2^l\cdot 11\in S_{5}$, $2^l\cdot 29\in S_{6}$. More precisely,
 $S_1=\{2^l\cdot 1\}$, $S_2=\{2^l\cdot 3\}$, $S_3=\{2^l\cdot 5\}$, $S_4=\{2^l\cdot 7 \}$, $S_5=\{2^l\cdot 11\}$, $S_6=\{2^l\cdot 15\}$ with $l=1,2,3,4,5$. It could be found that $f^l(c)$, $f^l(c')$, $-f^l(c)$, $-f^l(c')$, $ f^l(c)+f^l(c')$, $-(f^l(c)+f^l(c'))$ are distinct when $l$ run through $1,2,3,4,5$ and $i$ run through $0,1,2,\ldots,30$. It is not difficult to find that $\xi^i+2\xi^j$ are covered by coordinates of the first part of $H$, where $i\neq j$. The syndromes $\xi^i$, $2\xi^i$, $\xi^i+2\xi^i$ are covered by coordinates of the second part of $H$.
%By the meanings of equivalence \todo{???}, the code generated by  $H$ is uniquely determined.
\qed \end{proof}

\begin{remark} Note that $155=31\times 5\times 1$. And the size of every nonzero cyclotomic coset is $5$ since $5$ is a prime. On the other hand, $30=5 \times 6\times 1$. That means once we find a pair in the form of $\xi^{u_1}+2\xi^{u_2},\xi^{u_3}+2\xi^{u_4}$ and the sum of the pair is $\xi^{u_5}+2\xi^{u_6}$, and the opposites of the pair are respectively $\xi^{u_7}+2\xi^{u_8}$, $\xi^{u_9}+2\xi^{u_{10}}$, and the opposite of the sum of the pair is $\xi^{u_{11}}+2\xi^{u_{12}}$ such that $u_{2s}-u_{2s-1}$ with $s=1,2,\cdots,6$ exactly belong to six different cyclotomic cosets, respectively, then the check matrix is clear by the automorphism of $\mathrm{GR}(4^5)$.
\end{remark}

%=============================
%=============================
\subsection{\large\bf\boldmath A $1$-perfect code in $D(2667,0+127)$ (the case $\Gamma=0, \Delta=7$)}\label{TTT,SEC}
\begin{proposition}\label{p89}
Let $\xi$ be a primitive root of the basic irreducible polynomial $x^7+2x^4+x+3$ over $\Z_4$.
Let $H$ be a matrix which consists of $2667=127\times 21$ pairs of columns $(\xi^{2^li}+2\xi^{2^l(i+2)}, \xi^{2^l(i+2)}+2\xi^{2^l(i+7)})$, $(\xi^{2^li}+2\xi^{2^l(i+4)}, \xi^{2^l(i+2)}+2\xi^{2^l(i+17)})$, and $(\xi^{2^li}+2\xi^{2^l(i+10)}, \xi^{2^l(i+2)}+2\xi^{2^l(i+57)})$ with $i=0,1,2,\ldots,126$ and $l=0,1,2,3,4,5,6$ in the left part and $127$ columns  $\xi^i$, $i=0,1,2,\ldots,126$, in the right part.
The code $C$ defined by the check matrix $H$ is a $1$-perfect code in $D(2667,0+127)$.
\end{proposition}
\begin{proof}
The approach is the same as in the previous subsection. We outline the expression as follows.
Since $\xi$ is a primitive root of the basic irreducible polynomial $x^7+2x^4+x+3$ over $\Z_4$,
 we obtain that $\mathrm{GR}(4^7)=\mathcal{T}\oplus 2\mathcal{T}$ with  $\mathcal{T}=\{0,1,\xi,\xi^2,\ldots,\xi^{2^{7}-2}\}$.
Note that the size of every nonzero cyclotomic coset is $7$ since $7$ is a prime and $2667=127\times 21=127\times 7\times 3$, $2^7-2=7\times 18=7 \times 6\times 3$. It is sufficient to find coordinates of three pairs covering $\xi^{v_{2i-1}}+2\xi^{v_{2i}}$ with $i=1,2,\cdots,18$ such that $v_{2i}-v_{2i-1}$ exactly belong to $18$ distinct cyclotomic coset.

 In detail, we choose some $c$ and $c'$, as in the table below.
  $$\footnotesize
  \begin{array}{|c|c|c|c|c|c|}
    \hline
     c & c' & -c & -c' & c+c' & -(c+c') \\
    \hline &&&&& \\[-1.5ex]
     \xi^{i+1}+2\xi^{i+2}  &  \xi^{i+2}+2\xi^{i+7}  &  \xi^{i+1}+2\xi^{i+8}  &  \xi^{i+2}+2\xi^{i+56}  &  \xi^{i+8}+2\xi^{i+19}  &  \xi^{i+8}+2\xi^{i+95}  \\

     \xi^{i+1}+2\xi^{i+4}  &  \xi^{i+2}+2\xi^{i+17}  &  \xi^{i+1}+2\xi^{i+64}  &  \xi^{i+2}+2\xi^{i+33}  &  \xi^{i+8}+2\xi^{i+58}  &  \xi^{i+8}+2\xi^{i+91}  \\

     \xi^{i+1}+2\xi^{i+10}  &  \xi^{i+2}+2\xi^{i+57}  &  \xi^{i+1}+2\xi^{i+91}  &  \xi^{i+2}+2\xi^{i+15}  &  \xi^{i+8}+2\xi^{i+109}  &  \xi^{i+8}+2\xi^{i+92}  \\
    \hline
  \end{array}$$

  %we choose $\xi^{i+1}+2\xi^{i+2}$, $\xi^{i+2}+2\xi^{i+7}$ as a pair, then the opposite of them are $\xi^{i+1}+2\xi^{i+8}$, $\xi^{i+2}+2\xi^{i+56}$, respectively, and the sum of them is $\xi^{i+8}+2\xi^{i+19}$, and the opposite of their sum is $\xi^{i+8}+2\xi^{i+95}$. Similarly, choose $\xi^{i+1}+2\xi^{i+4}$, $\xi^{i+2}+2\xi^{i+17}$ as a pair, then $\xi^{i+1}+2\xi^{i+64}$, $\xi^{i+2}+2\xi^{i+33}$ are their opposite, and $\xi^{i+8}+2\xi^{i+58}$, $\xi^{i+8}+2\xi^{i+91}$ are their sum and opposite of them, respectively. Choose $\xi^{i+1}+2\xi^{i+10}$, $\xi^{i+2}+2\xi^{i+57}$ as a pair, then the opposite of them are $\xi^{i+1}+2\xi^{i+91}$, $\xi^{i+2}+2\xi^{i+15}$, respectively, and the sum of them is $\xi^{i+8}+2\xi^{i+109}$, and the opposite of their sum is $\xi^{i+8}+2\xi^{i+92}$. By listing the cyclotomic  coset modulo $127$, we have that
 %$S_1=\{1,2,4,8,16,32,64\}$ $S_2=\{3,6,12,24,48,96,65\}$ $S_3=\{5,10,20,40,80,33,66\} $S_4=\{7,14,28,56,112,97,67\}$ $S_5=\{9,18,36,72,17,34,68\}$ $S_6=\{11,22,44,88,49,98,69\}$ $S_7=\{13,26,52,104,81,35,70\}$ $S_8=\{15,30,60,120,113,99,71\}$ $S_9=\{19,38,76,25,50,100,73\}$ $S_{10}=\{21,42,84,41,82,37,74\}$ $S_{11}=\{23,46,92,57,114,101,75\}$ $S_{12}=\{27,54,108,89,51,102,77\}$ $S_{13}=\{29,58,116,105,83,39,78\}$ $S_{14}=\{31,62,124,121,115,103,79\}$ $S_{15}=\{43,86,45,90,53,106,85\}$ $S_{16}=\{47,94,61,122,117,107,87\}$ $S_{17}=\{55,110,93,59,118,109,91\}$ $S_{18}=\{63,126,125,123,119,111,95\}$
 Note that $1$, $5$, $7$, $54$, $11$, $87$, $3$, $15$, $63$, $31$, $50$, $83$, $9$, $55$, $90$, $13$, $101$, $84$ exactly belong to $18$ distinct cyclotomic cosets.
Then by the automorphism of $\mathrm{GR}(4^7)$ we naturally prove the statement.
 %Consequently, we choose $\xi^{l(i+1)}+2\xi^{l(i+2)}$, $\xi^{l(i+2)}+2\xi^{l(i+7)}$ as a pair, and $\xi^{l(i+1)}+2\xi^{l(i+4)}$, $ \xi^{l(i+2)}+2\xi^{l(i+17)}$ as another pair, and $\xi^{l(i+1)}+2\xi^{l(i+10)}$, $\xi^{l(i+2)}+2\xi^{l(i+57)}$ also as a pair with $i=0,1,2,\cdots,126$ and $l=1,2,4,8,16,32,64$. Then add these pairs into the left of $H$ over $\Z_4$. $H$, as the check matrix of $1$-perfect code in $D(2667,0+127)$, is constructed.
%By the meanings of equivalence  \todo{???}, the code generated by  $H$ is uniquely determined. At this time, we could view $1$-perfect code as a quasi-cyclic code of index $1\frac{1}{42}$ in the sense of equivalence.
\qed\end{proof}
%=============================
%=============================
\subsection{\large\bf Nonequivalence}\label{TTT,sub}
\begin{proposition}\label{p:noneq}
Each of the three quasi-cyclic codes in $D(7,0+7)$, $D(155,0+31)$, $D(2667,0+127)$ is not equivalent to the codes constructed in Section~\ref{ss:n''}, with the corresponding parameters.
\end{proposition}
\begin{proof} %\todo{extend the scetch below}.
Let us consider a quasi-cyclic code $C$ in $(m,0+n'')$, one of the three codes considered above, and a code $C'$ with the same parameters constructed in Subsection~\ref{sss:707} or Subsection~\ref{sss:n''gen}.
We first show that $C$ has only $n''(n''-1)/6$ codewords of weight $3$ that have zeros in the first $m$ coordinates, while $C'$ has more than $n''(n''-1)/6$.
Both codes have $n''(n''-1)/6$ order-$2$ codewords of weight $3$ with $0$s in the first part of coordinates (indeed, since the syndromes of order $2$  are covered by the last part of coordinates, the columns in the last part multiplied by $2$ are the all $n''$ columns of order $2$; there are exactly $n''(n''-1)/6$ triples of linearly dependent columns of order $2$).
The quasi-cyclic code $C$ have no codewords of weight $3$ by [\cite{Wan:4ary}, Proposition~9.8]
(this is also straightforward from the code distance of the ``Preparata'' codes constructed in \cite{HammonsOth:Z4_linearity}),
%\todo{show it with citing Sections 6.5 and 9.2 in [Wan. Quaternary codes] },
 while $C'$, by construction, has columns $a$, $b$, $a+b$ in the last part of the check matrix, which were added as in Section~\ref{ss:n''}.
%\todo{I agree that``The quasi-cyclic code $C$ have no weight-$3$ codewords", but I am confused about``$C$ has only $n''(n''-1)/6$ codewords of weight $3$ that have zeros in the first $m$ coordinates, while $C'$ has more than $n''(n''-1)/6$." Is it a contradiction? \\This means that $C$ and $C'$  are not equivalent, see below.}

Next, we consider an automorphism $\phi$ of the Doob graph $D(m,n'')$ that sends $C$ to $C'$. Without loss of generality, we can assume that  $\phi(0)=0$ (otherwise, we consider the automorphism $\phi'(\cdot)=\phi(\cdot)-\phi(0)$, which also sends $C$ to $C'=C'-\phi(0)$). It is easy to understand that any  automorphism $\phi$ such that $\phi(0)=0$
stabilizes the subgraph isomorphic to $D(0,n'')$ spanned by the last $n''$ coordinates (in particular, this subgraph contains all cliques of size $4$ containing $0$). This means that the number of weight-$3$ codewords in this subgraph is invariant for equivalent codes. Hence, $C$ and $C'$ are not equivalent.
\qed\end{proof}

%=============================
%=============================
%=============================
\section{\Large Conclusion and open problems}\label{TTT,OPEN}
 In this paper, we prove that the condition on the existence of additive perfect codes in Doob graphs given in \cite{Kro:perfect-doob} is necessary and sufficient  by constructing the check matrix of $1$-perfect code in $D(m,n'+n'')$, where $m$, $n'$, $n''$ satisfy the conditions (\ref{EQUA1}--\ref{EQUA3}) in Lemma \ref{1} with $\Gamma$ even and $\Delta$ odd and basing on  some known results in \cite{Kro:perfect-doob}. Meanwhile, we construct three quasi-cyclic additive $1$-perfect codes in  $D((2^{\Delta}-1)\frac{2^{\Delta}-2}{6},0+(2^{\Delta}-1))$ in the case of $\Gamma=0$ and $\Delta=3,5,7$, respectively. Constructing such class of $1$-perfect codes replies on a large number of calculations with increasing $\Gamma$ and $\Delta$. If some generalized approach could be proposed, it will be more meaningful. \begin{enumerate}
                      \item
 A natural question is: does there exist a $1$-perfect quasi-cyclic code of index $2^{\Delta}-1$ in $D((2^{\Delta}-1)\frac{2^{\Delta}-2}{6},0+(2^{\Delta}-1))$ for all odd prime $\Delta$, even for all odd $\Delta$?
 A similar question was considered in \cite{BorFer:1cyclic} for $\Z_2\Z_4$-cyclic $1$-perfect codes, which are also additive codes over the mixed $\Z_2\Z_4$ alphabet, but with the Lee metric.
 \item In Section~\ref{TTT,sub}, we established that there are at least two nonequivalent additive $1$-perfect codes for each of the parameters
 $(7,0+7)$, $(155,0+31)$, $(2667,0+127)$. It is expected that there are much more equivalence classes for these and other admissible parameters.
 In particular, we were able to find another additive $1$-perfect code in $D(7,0+7)$,
 different from the codes described in Sections~\ref{sss:707} and~\ref{TTT,707}:
   $$\left(\begin{array}{c@{\,}c@{\,\,\,}c@{\,}c@{\,\,\,}c@{\,}c@{\,\,\,}c@{\,}c@{\,\,\,}c@{\,}c@{\,\,\,}c@{\,}c@{\,\,\,}c@{\,}c@{\,}|@{\,}ccccccc}
 1&0& 2&2& 0&1& 1&1& 3&2& 1&2& 1&1&  1&0&0&3&0&1&1 \\
 0&1& 1&0& 2&2& 1&2& 1&1& 3&2& 1&0&  0&1&0&1&3&2&1 \\
 2&2& 0&1& 1&0& 3&2& 1&2& 1&1& 2&3&  0&0&1&0&1&1&1
\end{array}\right).$$
    A natural question is: how many nonequivalent additive $1$-perfect codes are there in $D(m,n'+n'')$ for any admissible $m$, $n'$, $n''$,
    in particular, in $D(7,0+7)$?
 Note that even in the case of $D(0,0+n)$, i.e., in the quaternary Hamming graph, the question is not easy:
 the existence on non-equivalent additive codes is connected with the existence on non-equivalent partitions of the additive group $\Z_2^{2n+}$ into subgroups isomorphic to $\Z_2^{2+}$ \cite{HerSho71}.
\end{enumerate}

\begin{remark} In our final remark, we briefly consider
the codes dual to the additive $1$-perfect codes.
In the case of the Hamming graphs $H(n,q)$
(including $D(0,n)=H(n,4)$), such codes, known as simplex codes,
have the cardinality $n(q-1)+1$ and the distance $(n(q-1)+1)/q$
between any two different codewords.
In an alternative notion, such objects are also known as tight $2$-designs.
As described in \cite{KoolMun:2000},
the codes dual to the additive $1$-perfect codes in $D(m,n)$
are also tight $2$-designs in the graph $D^*$ dual to $D(m,n)$.
This graph is built on the same group and has the same distance-regular parameters
as $D(m,n)$.
So, $D^*$ is isomorphic to $D(M,N)$ for some $M$ and $N$ such that $2M+N=2m+n$.
It is expected that $(M,N)=(m,n)$;
however, establishing the isomorphism needs some technique
and goes beyond the scope of the current research.
So, our results imply (assuming $(M,N)=(m,n)$ is true) the existence of the additive tight $2$-designs in $D(m,n)$ for the same $n$ and $m$ for which additive $1$-perfect codes are constructed, but the check matrix of an additive  $1$-perfect code
is not in general a generator matrix of a tight $2$-design in the same graph.
\end{remark}

\begin{acknowledgements}
The authors thank Tatsuro Ito, Jack Koolen, and Patrick Sol\'e for the consulting concerning the last remark and the anonymous referees for useful comments.
\end{acknowledgements}

%\bibliographystyle{spmpsci}
%\bibliography{../../k}   % name your BibTeX data base

\begin{thebibliography}{1}
\providecommand{\url}[1]{{#1}}
\providecommand{\urlprefix}{URL }
\expandafter\ifx\csname urlstyle\endcsname\relax
  \providecommand{\DOI}[1]{DOI~\discretionary{}{}{}#1}\else
  \providecommand{\DOI}{DOI~\discretionary{}{}{}\begingroup
  \urlstyle{rm}\Url}\fi

\bibitem{BorFer:1cyclic}
Borges, J., Fern\'{a}ndez-C\'{o}rdoba, C.: There is exactly one
  ${Z_2Z_4}$-cyclic $1$-perfect code,.
\newblock \href{http://link.springer.com/journal/10623}{Des. Codes
  Cryptography} \textbf{85}(557-566), 3 (2017).
\newblock \DOI{10.1007/s10623-016-0323-3}

\bibitem{HammonsOth:Z4_linearity}
Hammons Jr, A.R., Kumar, P.V., Calderbank, A.R., Sloane, N.J.A., Sol\'e, P.:
  The {$Z_4$}-linearity of {K}erdock, {P}reparata, {G}oethals, and related
  codes.
\newblock
  \href{http://ieeexplore.ieee.org/xpl/RecentIssue.jsp?punumber=18}{IEEE Trans.
  Inf. Theory} \textbf{40}(2), 301--319 (1994).
\newblock \DOI{10.1109/18.312154}

\bibitem{HedGuz:2015}
Heden, O., G\"uzeltepe, M.: On perfect $1$-{$\mathcal E$}-error-correcting
  codes.
\newblock \href{http://www.mathos.hr/mc/}{Math. Commun.} \textbf{20}(1), 23--35
  (2015).
%\newblock
  %\urlprefix\url{http://www.mathos.unios.hr/mc/index.php/mc/article/view/836}

\bibitem{HerSho71}
Herzog, M., Sch{\"o}nheim, J.: Linear and nonlinear single-error-correcting
  perfect mixed codes.
\newblock \href{http://www.sciencedirect.com/science/journal/00199958}{Inf.
  Control} \textbf{18}(4), 364--368 (1971).
\newblock \DOI{10.1016/S0019-9958(71)90464-5}

\bibitem{KoolMun:2000}
Koolen, J.H., Munemasa, A.: Tight $2$-designs and perfect $1$-codes in {D}oob
  graphs.
\newblock \href{http://www.sciencedirect.com/science/journal/03783758}{J. Stat.
  Plann. Inference} \textbf{86}(2), 505--513 (2000).
\newblock \DOI{10.1016/S0378-3758(99)00126-3}

\bibitem{Kro:perfect-doob}
Krotov, D.S.: Perfect codes in {D}oob graphs.
\newblock \href{http://link.springer.com/journal/10623}{Des. Codes
  Cryptography} \textbf{80}(1), 91--102 (2016).
\newblock \DOI{10.1007/s10623-015-0066-6}

\bibitem{Tiet:1973}
Tiet{\"a}v{\"a}inen, A.: On the nonexistence of perfect codes over finite
  fields.
\newblock \href{http://epubs.siam.org/journal/smjmap}{SIAM J. Appl. Math.}
  \textbf{24}(1), 88--96 (1973).
\newblock \DOI{10.1137/0124010}

\bibitem{Wan:4ary}
Wan, Z.X.: Quaternary Codes, \emph{Series on Applied Mathematics}, vol.~8.
\newblock World Scientific (1997).
\newblock \DOI{10.1142/3603}

\bibitem{ZL:1973}
Zinoviev, V., Leontiev, V.: The nonexistence of perfect codes over {G}alois
  fields.
\newblock Probl. Control Inf. Theory \textbf{2}(2), 123--132, 16--24[Engl.
  transl.] (1973)

\end{thebibliography}
%\end{document}

\providecommand\href[2]{#2} \providecommand\url[1]{\href{#1}{#1}}
  \def\DOI#1{{\small {DOI}:
  \href{http://dx.doi.org/#1}{#1}}}\def\DOIURL#1#2{{\small{DOI}:
  \href{http://dx.doi.org/#2}{#1}}}

\end{document}